\documentclass{llncs}
\usepackage{url}
\usepackage{array}
\usepackage{amssymb}
\usepackage{amsfonts}
\usepackage{amsmath}
\usepackage{mathrsfs}
\usepackage{latexsym}
\usepackage{graphicx}
\usepackage{subfigure}
\usepackage{psfrag}
\usepackage{enumerate}
\usepackage{url}
\usepackage{color}
\usepackage[T1]{fontenc}
\usepackage{verbatim}

\newtheorem{hypothesis}{Hypothesis}

\newenvironment{proofof}{\noindent \textit{Proof of claim. }}{\hfill$\Diamond$}

\def\is{\textsc{Independent Set}}
\def\m3sat{\textsc{Max-3Sat}}
\def\3sat{\textsc{3Sat}}
\def\ml{\textsc{Max-3Lin}}
\def\vc{\textsc{Vertex Cover}}
\def\ds{\textsc{Dominating Set}}
\def\mids{\textsc{Independent Dominating Set}}
\def\mcolo{\textsc{Coloring}}
\def\cb{\textsc{Bipartite Subgraph}}
\def\msp{\textsc{Set Packing}}
\newcommand{\minsat}{\textsc{Min-Sat}}
\def\pnp{$\mathbf{P}=\mathbf{NP}$}
\def\np{$\mathbf{NP}$}

\let\leq\leqslant
\let\geq\geqslant

\begin{document}

\title{\textbf{Subexponential and FPT-time Inapproximability of Independent Set and Related Problems}
\thanks{Research supported by the French Agency for Research under the DEFIS program TODO, ANR-09-EMER-010. A preliminary version of this work appeared in~\cite{escoffier-kim-paschos}.}}

\author{Bruno Escoffier\and Eun Jung Kim \and Vangelis~Th.~Paschos\thanks{Also, Institut Universitaire de France}}
\institute{PSL Research University, Universit\'{e} Paris-Dauphine, LAMSADE, CNRS UMR 7243\\
\texttt{\{escoffier,eun-jung.kim,paschos\}@lamsade.dauphine.fr}}

%\texttt{\{escoffier,eun-jung.kim,paschos\}@lamsade.dauphine.fr}}

\maketitle

\begin{abstract}
Fixed-parameter algorithms, approximation algorithms and moderately
exponential algorithms are three major approaches to algorithms
design. While each of them being very active in its own, there is an
increasing attention to the connection between different approaches.
In particular, whether \is{} would be better approximable once
endowed with subexponential-time or FPT-time is a central question.
In this paper, we present a strong link between the linear PCP
conjecture and the inapproximability, thus partially answering this
question.
\end{abstract}

\section{Introduction}\label{intro}

In this paper we look into three approaches to algorithms design: Fixed-pa\-ra\-me\-ter algorithms,
approximation algorithms and moderately exponential algorithms. These three areas, each of them being very active in its
own, have been considered as foreign to each other until recently. Polynomial-time approximation algorithm produces a solution whose quality is guaranteed to lie within a certain range from the optimum. One illustrative problem indicating the development of this area is \is{}.
%Given a graph~$G(V,E)$ on a set~$V$ of vertices and with a set~$E$ of edges, the goal of \is{} problem is to determine a maximum subset of~$V$ of pairwise disjoint vertices. This problem is one of the first~21 problems shown to be \np{}-complete in Karp's seminal paper~\cite{karpcl}.
The approximability of \is{} within
constant ratios\footnote{The approximation ratio of an algorithm
computing a feasible solution for some problem is the ratio of the
value of the solution computed over the optimal value for the
problem.} has remained as the most important open problems for a long time in the field. It was only after the novel
characterization of the \np{} given by the \textsc{PCP}
theorem~\cite{motwa,motwaj} that impossibility of such
approximability has been proven assuming \pnp{}. Subsequent
improvements of the original PCP theorem, leading to
corresponding refinements of the characterization of~\np{} have
also led to the actual very strong inapproximability result for
\is{}, namely, that it is inapproximable within
ratios~$\Omega(n^{\varepsilon - 1})$ for any $\varepsilon > 0$,
unless \pnp{}~\cite{zuckermanis06}.

Moderately exponential algorithm is to allow exponential running
time for the sake of optimality. In this case, the endeavor lies in
limiting the growth of running time function as slow as possible.
Parameterized complexity provides an alternative framework to
analyze the running time in a more refined way
\cite{downeybook,FG06}. The aim is to get an $O(f(k)\cdot
n^{c})$-time algorithm for some constant $c$ (independent of $k$).
As these two research programs offer a generous running time
compared to polynomial-time approximation algorithms, a growing
amount of attention is paid to them as a way to cope with hardness
in approximability. The first one deals with \textit{moderately
exponential approximation}. The goal of this program is to explore
approximability of highly inapproximable (in polynomial time)
problems in superpolynomial or moderately exponential time. Roughly
speaking, if a given problem is solvable in time say~$O^*(\gamma^n)$
but it is NP-hard to approximate within some ratio~$r$, we seek
$r$-approximation algorithms with complexity~- significantly~- lower
than $O^*(\gamma^n)$. This issue has been considered for several
problems such as \textsc{Set Cover}~\cite{cygan,effapproxsc},
\mcolo{}~\cite{incexcsetpart,effapproxcolor}, \is{} and
\vc{}~\cite{effapprox},
\textsc{Bandwidth}~\cite{CyganP10,FurerGK09}.

The second research
program handles approximation by fixed parameter
algorithms. In this approximation framework, we say that a
parameterized (with parameter~$k$) problem~$\Pi$ is $r$-approximable
if there exists an algorithm taking as inputs an instance~$I$
of~$\Pi$ and~$k$ and either computes a solution smaller or greater
than (depending on whether~$\Pi$ is, respectively, a minimization,
or a maximization problem)~$rk$, or returns ``no'', asserting in
this case that there is no solution of value at most or at least
$k$. This line of research was initiated by three independent works~\cite{dofemcciwpec,caihuiwpec,ChenGG06}. As an excellent overview in this direction, see~\cite{marx-approx}.

Several natural questions can be asked dealing with these two
programs. In particular, the following ones have been asked several
times (see for
instance~\cite{marx-approx,dofemcciwpec,FurerGK09,effapprox}) and of great interest:
\begin{enumerate}
\item[{\bf Q1}]\label{ena} can a highly inapproximable in polynomial time problem be well-ap\-p\-ro\-x\-i\-ma\-ted in subexponential time?
\item[{\bf Q2}]\label{dyo} does a highly inapproximable in polynomial time problem become well-approximable in parameterized time?
%\item[{\bf Q3}]\label{tria} is a W[$\cdot$]-hard problem well-approximable in parameterized time?
\end{enumerate}
Few answers have been obtained until now. Regarding {\bf Q1},
negative results can be directly obtained by gap-reductions for
certain problems. For instance, \mcolo{} is not approximable within
ratio $4/3-\epsilon$, since this would allow to determine whether a
graph is 3-colorable or not in subexponential time. This contradicts
a widely-acknowledge computational assumption~\cite{impa}:
\begin{quote}
\textit{Exponential Time Hypothesis} (\textsc{ETH}): There exists an $\epsilon > 0$ such that no algorithm
solves \3sat{} in time $2^{\epsilon n}$, where $n$ is the number
of variables.
\end{quote}
Regarding {\bf Q2},~\cite{dofemcciwpec} shows that assuming FPT $\neq$ W[2], for any~$r$ the \mids{} problem is not $r$-approximable\footnote{Actually, the result is even stronger: it is impossible to obtain a ratio $r=g(k)$ for any function $g$.} (in FPT time).
%As to {\bf Q3}, there exist problems for which this question can be positively answered. This is, for instance, the case of \textsc{max $k$-coverage} and \textsc{max $(k,n-k)$-cut} that are W[1]-hard with respect to parameter~$k$~\cite{cai}, but they admit polynomial time approximation schemata (i.e., algorithms achieving ratios $1-\epsilon$, for any $\epsilon > 0$),see~\cite{marx-approx,bonnet-escoffier-paschos-tourniaire}, respectively.

Among interesting problems for which {\bf
Q1} and {\bf Q2} are worth being asked are \is{}, \mcolo{} and \ds{}. They fit in the frame of both {\bf Q1} and {\bf Q2} above: they are hard to approximate in polynomial time
while their approximability in subexponential or in parameterized
time is still open. Note that \is{} and \ds{} are moderately
exponential approximable within any ratio $1 - \varepsilon$, for any
$\varepsilon > 0$~\cite{effapproxsc,effapprox}, while \mcolo{} is
approximable within ratio $(1 + 1/\chi(G))$, where~$\chi(G))$
denotes the chromatic number of a graph~$G$ in moderately
exponential time~\cite{incexcsetpart,effapproxcolor}.
%Let us note that, as it is proved in~\cite{effapprox}, \is{} is approximable in subexponential time within any ratio~$\Omega(n^{-1})$ (that, as we have mentioned just above, is unachievable in polynomial time).

Our contribution in this paper is to establish a link between a
major conjecture in PCP theorem and inapproximability in
subexponential-time and in FPT-time, assuming ETH. We first state
the conjecture while the definition of PCP is deferred to the next
section.
%\textcolor{blue}
{\begin{quote}
\textit{Linear PCP Conjecture} (\textsc{LPC}): \3sat{} $\in
{\textsc{PCP}}_{1,1/2}[\log{|\phi|} + D, E]$, where $|\phi|$ is the
size of the \3sat{} instance (sum of lengths of clauses), $D$ and
$E$ are constant.
\end{quote}}
Unlike \textsc{ETH} which is arguably recognized as a valid
statement, \textsc{LPC} is a wide open question. In the emphasized
statement given just below, we claim that if \textsc{LPC} turns out
to hold, it immediately implies that one of the most interesting
questions in subexponential and parameterized approximation is
negatively answered. In particular, as shown in the sequel, assuming \textsc{ETH} the followings hold for \is{} on $n$ vertices, for any constant $0<r<1$: \\
%\begin{enumerate}
%\item[(i)]
(i)~There is no $r$-approximation algorithm in time $O(2^{n^{1-\delta}})$ for any $\delta>0$. \\
%\item[(ii)]
(ii)~There is no $r$-approximation algorithm in time $O(2^{o(n)})$ if \textsc{LPC} holds. \\
%\item[(iii)]
(iii)~There is no $r$-approximation algorithm in time $O(f(k)n^{O(1)})$ if \textsc{LPC} holds.
%\end{enumerate}

%\end{theorem}

Remark that (i) is not conditional upon \textsc{LPC}. In fact, this
is an immediate consequence of near-linear PCP construction
%\textcolor{blue}
{achieved in~\cite{Dinur07}. Note that similar inapproximability
results under \textsc{ETH} for \m3sat{} and \ml{} for some
subexponential running time
 have been obtained in
\cite{DBLP:conf/focs/MoshkovitzR08}.}

In the following, Section~\ref{sec:pre} reviews some known
consequences of near-linear PCP. In Section~\ref{secparam}, we show
how a combination of two classic reductions yields parameterized
inapproximabiliy bounds for \is{} provided that {\textsc LPC} and
{\bf ETH} hold (point (iii) above); we also provide a parameterized
approximation preserving reduction that allows to transfer
parameterized inapproximability results to \ds{}. In
Section~\ref{secsubexpo}, we analyze known reductions in the view of
inapproximability in subexponential running time and present some
results similar to (i) and (ii).

%In Section~\ref{secsubexpo} further consequences of the {\textsc{PCP}} theorem in~\cite{DBLP:conf/focs/MoshkovitzR08}. The first one is that, under {{{\textsc{ETH}}}}, \is{} is not approximable within any constant ratio in time~$2^{n^{1-\delta}}$, for any $\delta > 0$, where~$n$ is the size of the input graph.
%Note that \mcolo{} is approximable within ratio $(1 +
%1/\chi(G))$, where~$\chi(G))$ denotes the chromatic number of a
%graph~$G$ in moderately exponential
%time~\cite{incexcsetpart,effapproxcolor}.

%First, we show that using a slightly stronger hypothesis (a ``stronger version'' of a negative result based on \textsc{ETH}), there is no $(1/2 + \epsilon)$-approximation algorithm running in time~$f(k)n^{o(k)}$ (for any function $f$) for \is{}. Furthermore, using this result, we prove that, under the same hypothesis, for any $r \in (0,1)$, no parameterized algorithm can achieve approximation ratio~$r$ for \is{}. Finally, using a parameterized reduction that preserves some level of approximability, we show that, always under the same hypothesis, no $(3/2 - \epsilon)$-approximation algorithm exists running in time~$f(k)n^{o(k)}$ (for any function $f$) for \ds{}.

\section{Preliminaries}\label{sec:pre}

%[TO BE DETAILED]
\subsection{PCP and inapproximability of \m3sat{}}

A problem is in ${\textsc{PCP}}_{\alpha,\beta}[q,p]$ if there exists
a PCP verifier which uses $q$ random bits, reads at most $p$ bits in
the proof and is such that:
\begin{itemize}
\item if the instance is positive, then there exists a proof such that V(erifier) accepts with probability at least $\alpha$;
\item if the instance is negative, then for any proof~V accepts with probability at most~$\beta$.
\end{itemize}
%\textcolor{blue}
{Based upon the above definition, the following
theorem is proved in~\cite{Dinur07} (see also Theorem 7 in
\cite{DBLP:conf/focs/MoshkovitzR08}), presenting a further
refinement of the characterization of NP.
\begin{theorem}\cite{Dinur07}\label{2-query}
For every $\epsilon>0$,
$$
\3sat{} \in {\textsc{PCP}}_{1,\epsilon}[(1+o(1))\log n
+O(\log(1/\epsilon)),O(\log(1/\epsilon))]
$$
\end{theorem}
A recent improvement \cite{DBLP:conf/focs/MoshkovitzR08} of
Theorem~\ref{2-query} (a PCP Theorem with two-query projection
tests, sub-constant error and almost-linear size) has some important
corollaries in polynomial approximation. Among those, the following
two are of particular interest in what follows.}
%\textcolor{blue}{NB: SHALL WE STATE THE EXACT THEOREM (TH 9) IN \cite{DBLP:conf/focs/MoshkovitzR08}? I DO NOT THINK SO}
\begin{corollary}\cite{DBLP:conf/focs/MoshkovitzR08}\label{2-query-cor1}
Under \textsc{ETH}, for every $\epsilon>0$, and $\delta>0$, it is
impossible to distinguish between instances of \textsc{max 3-lin}
with $m$ equations where at least $(1-\epsilon)m$ are satisfiable
from instances where at most $(1/2+\epsilon)m$ are satisfiable, in
time $O(2^{m^{1-\delta}})$.
\end{corollary}
\begin{corollary}\cite{DBLP:conf/focs/MoshkovitzR08}\label{2-query-cor2}
Under \textsc{ETH}, for every $\epsilon>0$, and $\delta>0$, it is
impossible to distinguish between instances of \m3sat{}
with $m$ clauses where at least $(1-\epsilon)m$ are satisfiable
from instances where at most $(7/8+\epsilon)m$ are satisfiable, in
time $O(2^{m^{1-\delta}})$.
\end{corollary}
%The verifier works as follows: it takes a random bits $R$, reads $E$ bits in the proof and based on these $E$ bits (and $R$) either accepts or not.
The following is a stronger version of Corollary \ref{2-query-cor2}: it holds if {\textsc LPC} holds. This will be our working
hypothesis.
\begin{hypothesis}
\label{hyp0} Under \textsc{ETH}, there exists $r<1$ such that: for every
$\epsilon > 0$ it is impossible to distinguish between instances of
\m3sat{} with $m$ clauses where at least $(1-\epsilon)m$ are
satisfiable from instances where at most $(r + \epsilon)m$ are
satisfiable, in time $2^{o(m)}$.
\end{hypothesis}

Using the well known sparsification lemma
(Lemma~\ref{sparsificationlemma}), which intuitively allows to work
with 3-SAT formula with linear lengths (the sum of the lengths of
clauses is linearly bounded in the number of variables), a very
standard argument gives the validity of Hypothesis~\ref{hyp0} under
{\textsc LPC}, see Lemma~\ref{lpc2hyp}.
\begin{lemma}\cite{impa}\label{sparsificationlemma}
For all $\epsilon>0$, a 3-SAT formula $\phi$ on $n$ variables can be
written as the disjunction of at most $2^{\epsilon n}$ 3-SAT formula
$\phi_i$ on (at most) $n$ variables such that $\phi_i$ contains each
variable in at most $c_\epsilon$ clauses for some function
$c_\epsilon$. Moreover, this reduction takes at most $p(n)
2^{\epsilon n}$ time.
\end{lemma}

\begin{lemma}\label{lpc2hyp}
If {\textsc LPC}\footnote{Note that {\textsc LPC} as expressed in
this article implies that Hypothesis~\ref{hyp0} holds event with
replacing $(1-\epsilon)m$ by $m$. However, we define
Hypothesis~\ref{hyp0} with this lighter statement $(1-\epsilon)m$ in
order, in particular, to emphasize the fact that perfect
completeness is not required in the LPC conjecture.} holds, then
Hypothesis~\ref{hyp0} also holds.
\end{lemma}
\begin{proof}
{Suppose that $\3sat{} \in {\textsc{PCP}}_{1,1/2}[\log{|\phi|} + D,
E]$, where $|\phi|$ is the sum of the lengths of clauses  in the
\3sat{} instance, $D$ and $E$ are constants.}

%\textcolor{blue}
{Given an $\epsilon>0$, let $\epsilon'$ such that
$0<\epsilon'<\epsilon$.
 Given an instance~$\phi$ of
\textsc{3 sat} on $n$ variables, we apply the sparsification lemma
 (with $\epsilon'$) to get
$2^{\epsilon' n}$ instances $\phi_i$ on at most $n$ variables. Since
each variable appears at most $c_{\epsilon'}$ times in $\phi_i$, the
global size of $\phi_i$ is $|\phi_i|\leq c_{\epsilon'}n$.}

%\textcolor{blue}
{Then for each formula $\phi_i$ we use the previous PCP assumption.
The size of the proof is at most $E2^{|R|}=c'|\phi_i|\leq cn$ for
some constants $c',c$ that depend on $\epsilon'$ (where
$|R|=\log{n}+D$ is the number of random bits) since $E2^{|R|}$ is
the total number of bits that we read in the proof. Take one
variable for each bit in the proof: $x_1,\cdots,x_{cn}$. For each
random string $R$: take all the $2^E$ possibilities for the $E$
variables read, and write a CNF formula which is satisfied if and
only if the verifier accepts. This can be done with a formula with a
constant number of clauses, say $C_1$, each clause having a constant
number of variables, say $C_2$ ($C_1$ and $C_2$ depends on $E$).}

%\textcolor{blue}
{If we consider the CNF formed by all theses CNF for all the random
clauses, we get a CNF with $C_12^{|R|}$ clauses on variables
$x_1,\cdots,x_{cn}$. The clauses are on $C_2$ variables but by
adding a constant number of variables we can replace a clause on
$C_2$ variables by an equivalent set of clauses on 3 variables. This
way we get a 3-CNF formula and multiply the number of variables and
the number of clauses by a constant, so they are still linear in
$n$. For each $R$ you have a set of say $C'_1$ clauses.}

%\textcolor{blue}
{Suppose that we start from a satisfiable formula $\phi_i$. Then
there exists a proof for which the verifier always accepts. By
taking the corresponding values for the variables $x_i$, and
extending it properly to the new variables y, all the clauses are
satisfied.}

%\textcolor{blue}
{Suppose that we start from a non satisfiable formula $\phi_i$. Then
for any proof (i.e. any truth values of variables), the verifier
rejects for at least half of the random strings. If the verifier
rejects for a random string $R$, then in the set of clauses
corresponding to this variable at least one clause is not satisfied.
It means that among the $C'_12^{|R|}$ clauses (total number of
clauses), at least $1/2\cdot 2^{|R|}$ are not satisfied, ie a
fraction $1/(2C'_1)$ of the clauses.}

%\textcolor{blue}
{Then either $m=C'_12^{|R|}=O(n)$ clauses are satisfiable, or at
least $m/(2C'_1)$ clauses are not satisfied by each assignment.
Distinguishing between these sets in time $2^{o(m)}$ would determine
whether $\phi_i$ is satisfiable or not in $2^{o(n)}$. Doing this for
each $\phi_i$ would solve \3sat{} in time
$p(n)2^{\epsilon'n}+2^{\epsilon' n}O(2^{o(n)})=O(2^{\epsilon n})$
(where $p$ is a polynomial). This is valid for any $\epsilon>0$ so
it would contradicting \textsc{ETH}.}\qed\end{proof}

%\textcolor{blue}{SHALL WE ADD THE FOLLOWING PRECISION? IN FOOTNOTE?
%MAYBE NOT? Note that {\textsc LPC} as expressed in this article
%implies that Hypothesis~\ref{hyp0} holds event with replacing
%$(1-\epsilon)m$ by $m$. However, we define Hypothesis~\ref{hyp0}
%with this lighter statement $(1-\epsilon)m$ in order in particular
%to emphasize the fact that perfect completeness is not required in
%the LPC conjecture.}
%
Dealing with \is{}, it is easy to see that, for any increasing and
unbounded function $r(n)$, the problem is approximable within
ratio~$1/r(n)$ in subexponential time (recall that ratios~$n^{\epsilon-1}$ are are very unlikely to be achieved in
polynomial time). Indeed, simply consider all the subsets of~$V$
of size at most~$n/r(n)$ and return the largest independent set
among these sets. If a maximum independent set has size at most~$n/r(n)$ then the algorithm finds it, otherwise the algorithm
outputs a solution of size~$n/r(n)$, while the size of an optimum
solution is at most~$n$. The running time of the algorithm
is~$O^*({n \choose n/r(n)})$ that is subexponential in~$n$.

Let us note that \is{} has the so called self-improvement
property~\cite{gj} claiming, roughly speaking, that either it is
polynomially approximable by a polynomial time approximation schema,
or no polynomial algorithm exists that guarantees some constant
approximation ratio, unless \pnp{}.

With a similar proof, the above self-improvement property can be
proved for \is{} also in the case of parameterized approximation.
\begin{lemma}\cite{escoffier-paschos-tourniaire}\label{selfimp}
The following statements are equivalent for \is{}:
\begin{itemize}
\item there exists $r\in (0,1)$ such that there exists an $r$-approximation parameterized algorithm;
\item for any $r\in (0,1)$ there exists an $r$-approximation parameterized algorithm.
    \end{itemize}
\end{lemma}
%\textcolor{blue}

%An $r$-approximation parameterized algorithm is an algorithm that, given a graph $G$ and an integer $k$, in time $f(k)p(n)$, either outputs a solution of size at least $rk$, or asserts that there is no solution of size at least $k$.

\subsection{Expander Graphs}

\begin{definition}
A graph $G$ is a $(n,d,\alpha)$-expander graph if (i) $G$ has $n$
vertices, $(ii)$ $G$ is d-regular, $(iii)$ all the eigenvalues
$\lambda$ of $G$ but the largest one is such that $|\lambda|\leq
\alpha d$.
\end{definition}
{\bf Fact 1.} {\it For any $k\in \mathbb{N}^*$ and any $\alpha>0$
there exists $d$ and a $(k^2,d,\alpha)$-expander graph. Moreover,
$d$ depends only on $\alpha$, and this graph can be
computed in polynomial time for every fixed $\alpha$.}\\

This fact follows from the following lemmas.
\begin{lemma}[\cite{gg}, or Th. 8.1 in~\cite{survey}]
For every positive integer $k$, there exists a
$(k^2,8,5\sqrt{2}/8)$-expander graph, computable in polynomial time.
\end{lemma}
If $G$ is a graph with adjacency matrix $M$, let us denote $G^k$ the
graph with adjacency matrix $M^k$.
\begin{lemma}[Fact 1.2 in~\cite{exp}]
If $G$ is a $(n,d,\alpha)$-expander graph, then $G^k$ is a
$(n,d^k,\alpha^k)$-expander graph.
\end{lemma}
{\it Proof.} $G^k$ is obviously $d^k$ regular, and the eigenvalues
of $G^k$ are the eigenvalues of $G$ to the power of $k$.~\qed

\noindent {\it Proof of Fact 1.} Take $\alpha>0$ and let $p$ be the
smallest integer such that $(5\sqrt{2}/8)^{p}\leq \alpha$. $G^{p}$
is as required. The proof of Fact~1 is completed.~\qed

Let $G$ be a graph on $n$ vertices and $H$ be a
$(n,d,\alpha)$-expander graph. Let $t$ be a positive integer. We
build the graph $G'_t$ on $N=nd^{t-1}$ vertices: each vertex
corresponds to a $(t-1)$-random walk $x=(x_1,\cdots,x_t)$ on $H$
(meaning that $x_1$ is chosen at random, and $x_{i+1}$ is chosen
randomly in the set of neighbors of $x_i$), and two vertices
$x=(x_1,\cdots,x_t)$ and $y=(y_1,\cdots,y_t)$ in $G'_t$ are adjacent
iff $\{x_1,\cdots,x_t,y_1,\cdots,y_t\}$ is a clique in $G$.
\begin{theorem}[claims 3.15 and 3.16 in~\cite{survey}]\label{th1}
Let $G$ be a graph on $n$ vertices and $H$ be a
$(n,d,\alpha)$-expander graph. If $b>6\alpha$, then:
\begin{itemize}
    \item If $\omega(G)\leq bn$ then $\omega(G'_t)\leq (b+2\alpha)^t
    N$;
    \item If $\omega(G)\geq bn$ then $\omega(G'_t)\geq (b-2\alpha)^t
    N$.
\end{itemize}
\end{theorem}
We are now able to prove the gap amplification with linear size
amplification.
\begin{theorem}\label{theoappendix}
Let $G$ be a graph on $n$ vertices (for a sufficiently large $n$)
and $a>b$ be two positive real numbers. Then for any real $r>0$ one
can build in polynomial time a graph $G_r$ such that:
\begin{itemize}
    \item $G_r$ has $N\leq Cn$ vertices for $C$ independent of $G$
    ($C$ may depend on $r$);
    \item If $\omega(G)\leq bn$ then $\omega(G_r)\leq b_rN$;
    \item If $\omega(G)\geq an$ then $\omega(G_r)\geq a_rN$;
    \item $b_r/a_r\leq r$.
\end{itemize}
\end{theorem}
{\it Proof.} Let $k=\lceil \sqrt{n}\rceil$. We modify $G$ by adding
$k^2-n$ dummy (isolated) vertices. Let $G'$ be the new graph. It has
$n'=k^2$ vertices. Note that $n'\leq
(\sqrt{n}+1)^2=n+2\sqrt{n}+1=n+o(n)$. Let $n$ be such that
$1-\epsilon\leq n/n'\leq 1$ for a small $\epsilon$. Thanks to Fact
1, we consider a $(k^2,d,\alpha)$-expander graph $H$ for a
sufficiently small $\alpha$ (the value of which will be fixed
later). According to Theorem~\ref{th1} (applied on $G'$) we build in
polynomial time a graph $G'_t$ on $N=n'd^t$ vertices such that
(choosing $\alpha<b/6$):
\begin{itemize}
    \item If $\omega(G)\leq bn$ then $\omega(G')=\omega(G)\leq bn'$, hence $\omega(G'_t)\leq (b+2\alpha)^t
    N$;
    \item If $\omega(G)\geq an$ then $\omega(G')=\omega(G)\leq an'(1-\epsilon)$, hence
    $\omega(G'_t)\geq (a(1-\epsilon)-2\alpha)^tN$.
\end{itemize}
We choose $\epsilon$ and $\alpha$ such that
$a(1-\epsilon)-2\alpha>b+2\alpha$, and then $t$ such that
$(a(1-\epsilon)-2\alpha)^t/(b+2\alpha)^t\leq r$. The number of
vertices of $G'_t$ is clearly linear in $n$ (first point of the
theorem). $b_r=(b+2\alpha)^t$ and $a_r=(a(1-\epsilon)-2\alpha)^t$
fulfills items 2, 3 and 4.~\qed

\section{Parameterized inapproximability bounds}\label{secparam}

It is shown in~\cite{chen} that, under \textsc{ETH}, for any
function $f$ no algorithm running in time $f(k)n^{o(k)}$ can
determine whether there exists an independent set of size $k$, or
not (in a graph with $n$ vertices). A challenging question is to
obtain a similar result for approximation algorithms for \is{}. In
the sequel, we propose a reduction from \m3sat{} to \is{} that,
based upon the negative result of Corollary~\ref{2-query-cor2}, only
gives a negative result for {\it some} function $f$ (because
Corollary~\ref{2-query-cor2} only avoids {\it some} subexponential
running time). However, this reduction gives the desired
inapproximability result if Hypothesis~\ref{hyp0}, which is an
enforcement of Corollary~\ref{2-query-cor2}, is used.

%\hypt\label{hyp1} For every $\epsilon>0$, it is impossible to distinguish between instances of \ml{} with $m$ equations where at least $(1-\epsilon)m$ are satisfiable, from instances with at most $(1/2+\epsilon)m$ satisfiable equations in time $2^{o(m)}$.~\ehypt
%\hypt\label{hyp2} For every $\epsilon>0$, and $\delta>0$, it is impossible to distinguish between instances of \m3sat{} with $m$ clauses where at least $(1-\epsilon)m$ are satisfiable, from instances where at most $(7/8+\epsilon)m$ are satisfiable,  in time $2^{o(m)}$.~\ehypt
Based upon Hypothesis~\ref{hyp0}, the
following theorem on parameterized inapproximability bound can be proved. Its proof essentially combines the parameterized reduction in~\cite{chen} and a classic gap-creating reduction.
%\footnote{We thank an anonymous referee for pointing out that Theorem~\ref{paraminapis} was proved in unpublished works of several independent people.}.
%Recall that Hypothesis~\ref{hyp0} assumes \textsc{ETH}.
\begin{theorem}\label{paraminapis}
Under Hypothesis~\ref{hyp0} and \textsc{ETH}, for every
$\epsilon>0$, no parameterized approximation algorithm for \is{}
running in time $f(k)N^{o(k)}$ can achieve approximation ratio
$r+\epsilon$ in graphs of order~$N$.
\end{theorem}
\begin{proof} Suppose that such an algorithm exists for some $\epsilon>0$.
W.l.o.g., we can assume that $f$ is increasing, and that $f(k)\geq
2^k$. Take an instance $I$ of \m3sat{}, let $K$ be an integer that
will be fixed later, and do the following:
%\begin{itemize}
%\item
Partition the $m$ clauses into $K$ groups $H_1,\cdots,H_K$ each of them containing, roughly, $m/K$ clauses each.
%\item
Each group $H_i$ involves a number $s_i\leq 3m/K$ of variables. For all possible values of these variables,
add a vertex in the graph $G_I$ if these values satisfy at least
$\lambda m/K$ clauses in $H_i$ (the value of $\lambda$ will also
be fixed later).
%\item
Finally, add an edge between two vertices if they have one contradicting variable.
%\end{itemize}
In particular the vertices corresponding to the same group of
clauses form a clique. It is easy to see that the so-constructed
graph contains $N\leq K2^{3m/K}$ vertices.

The following easy claim holds.
\begin{claim}\label{followinglemma}
If a variable assignment satisfies at least $\lambda m/K$
clauses in at most~$s$ groups, then it satisfies at most
$\lambda m+s(1-\lambda)m/K$ clauses.
\end{claim}
\begin{proofof} Consider an assignment as
the one claimed in claim's statement. This assignment satisfies at
most~$m/K$ clauses in at most~$s$ groups, and at most~$\lambda m/K$
in the other $K-s$ groups, so in total at most $sm/K+(K-s)\lambda
m/K=\lambda m+s(1-\lambda)m/K$, that completes the proof of the claim.\end{proofof}

Now, let us go back to the proof of the theorem.
%\begin{itemize}
%\item
Assume an independent set of size at least $t$ in $G_I$. Then one can achieve a partial solution that satisfies at least $\lambda m/K$
clauses in at least $t$ groups. So, at least $t\lambda m/K$ clauses
are satisfiable. In other words, if at most $(r+\epsilon')m$ clauses
are satisfiable, then a maximum independent set in~$G_I$ has size at
most $K\frac{r+\epsilon'}{\lambda}$.
%\item
Suppose that at least $(1-\epsilon')m$ clauses are satisfiable. Then, using Lemma~\ref{followinglemma},
there exists a solution satisfying at least $\lambda m/K$
clauses in at least $\frac{1-\epsilon'-\lambda}{1-\lambda}K$
groups; otherwise, it should be $\lambda
m+s(1-\lambda)m/K<(1-\epsilon')m$. Then, there exists an
independent set of size $\frac{1-\epsilon'-\lambda}{1-\lambda}K$
in $G_I$.
%\end{itemize}

Now, set $K=\lceil \phi(m)/(1-\epsilon^2)\rceil$ where $\phi$ is the
inverse function of $f$ (i.e., $\phi = f^{-1}$). Set also
$\lambda=1-\epsilon$, and $\epsilon'=\epsilon^3$. Run the assumed
$(r+\epsilon)$-approximation parameterized algorithm for \is{} in
$G_I$ with parameter $k=(1-\epsilon^2)K$. Then,
%\begin{itemize}
%\item
if at least $(1-\epsilon')m$ equations are satisfiable, there exists an independent set of size at least
$\frac{1-\epsilon'-\lambda}{1-\lambda}K=(1-\epsilon^3/\epsilon)K=(1-\epsilon^2)K=k$;
so, the algorithm must output an independent set of size at least
$(r+\epsilon)k$.
%\item
Otherwise, if at most $(r+\epsilon')$ equations are satisfiable, the size of an independent set is at most
$K\frac{r+\epsilon'}{\lambda}=K\frac{r+\epsilon^3}{1-\epsilon}=k\frac{r+\epsilon^3}{(1-\epsilon)(1-\epsilon^2)}=k(r+r\epsilon+o(\epsilon))$.
%\end{itemize}
%$(1/2+\epsilon)k=(1/2+\epsilon)(1-\epsilon^2)K=(1/2+\epsilon+o(\epsilon))K$,

So, for $\epsilon$ sufficiently small, the algorithm allows to
distinguish between the two cases of \m3sat{} (for $\epsilon'$).

%The running time of the yielded algorithm is $f(k)\cdot N^{k^{1-\delta}}$. We have
%\begin{eqnarray}
%f(k)\cdot N^{k^{1-\delta}}& = & m\cdot (K2^{3m/K})^{k^{1-\delta}} \qquad \because f(k)=m, N\leq K2^{3m/K}\\
%                        &\leq & m\cdot 2^{\frac{3(1-\epsilon^2)m}{\phi(m)^{\delta}}} \qquad \because K=\frac{\phi(m)}{1-\epsilon^2}, \phi(m)\leq \log{m}, \log{m}^{\log{m}}\leq m \\
%                       & \leq & Cm \cdot 2^{\frac{m}{m^{\delta'}}} \qquad \because \exists \delta' m^{\delta'}\leq \phi(m)^{\delta}
%\end{eqnarray}
%The last $\because$ can be satisfied by taking $\delta'<\delta\cdot \frac{\log{\phi(m)}}{\log{m}}$. This is possible since $\lim_{m\rightarrow \infty}\frac{\log{\phi(m)}}{\log{m}}=\infty$.

The running time of the yielded algorithm is $f(k)N^{o(k)}$, but $f(k)=f((1-\epsilon^2)K)=m$, and $N^{o(k)}=N^{k/\psi(k)}$ for some increasing and unbounded function~$\psi$, and $N^{o(k)}=
(K2^{3m/K})^{k/\psi(k)}=2^{o(m)}$.~\qed \end{proof}
%Note that a similar result (but with a weaker gap of $7/8+\epsilon$) can be obtained using Hypothesis~\ref{hyp2} instead.

Using Lemma~\ref{selfimp} together with Theorem~\ref{paraminapis}, the following result can be easily derived.
\begin{corollary}
Under Hypothesis~\ref{hyp0} and \textsc{ETH}, for any $r\in(0,1)$
there is no $r$-ap\-p\-ro\-x\-i\-ma\-ti\-on parameterized algorithm
for \is{}(i.e., an algorithm that runs in time $f(k)p(n)$ for some
function $f$ and some polynomial~$p$).
\end{corollary}
Let us now deal with \ds{}  that is known to be
W[2]-hard~\cite{downeybook}. Existence of FPT-approximation algorithms for this problem is
an open question~\cite{dofemcciwpec}. Here, we present an
approximation preserving reduction (fitting the parameterized
framework) that works with the special set of instances produced
in the proof of Theorem~\ref{paraminapis}. This reduction will
allow us to obtain a lower bound (based on the same hypothesis)
for the approximation of \textsc{min dominating set} from
Theorem~\ref{paraminapis}.

Consider a graph $G(V,E)$ on $n$ vertices where $V$ is a set of
$K$ cliques $C_1,\cdots,C_K$. We build a graph $G'(V',E')$ such
that $G$ has an independent set of size $\alpha$ if and only if
$G'$ has a dominating set of size $2K-\alpha$. The graph~$G'$ is
built as follows.
%\begin{itemize}
%\item
For each clique $C_i$ in $G$, add a clique $C'_i$ of the same size in $G'$. Add also:
%\begin{itemize} \item
an independent set $S_i$ of size $3K$, each vertex in $S_i$ being adjacent to all vertices in $C'_i$ and a special vertex $t_i$ adjacent to all the vertices in $C'_i$.
%\end{itemize}
%\item
For each edge $e=(u,v)$ with $u$ and $v$ {\it not} in the same clique in $G$, add an independent set~$W_e$ of size $3K$. Suppose that $u\in C_i$ and $v\in C_j$. Then, each vertex in $W_e$ is linked to $t_i$ and to all vertices in $C'_i$ but $u$ (and $t_j$ and all vertices in $C'_j$ but $v$).
%\end{itemize}

Informally, the reduction works as follows. The set $S_i$ ensures
that we have to take at least one vertex in each $C'_i$, the fact
that $|W_e| = 3K$ ensures that it is never interesting to take a
vertex in $W_e$. If we take vertex $t_i$ in a dominating set, this
will mean that we do not take any vertex in the set $C_i$ in the
corresponding independent set in $G$. If we take one vertex in
$C'_i$ (but not $t_i$), this vertex will be in the independent set
in $G$. Let us state this property in the following lemma.
\begin{lemma}\label{lemmads}
    $G$ has an independent set of size $\alpha$ if and only if $G'$
    has a dominating set of size $2K-\alpha$.
\end{lemma}
\begin{proof}
 Suppose that $G$ has
an independent set $S$ of size $\alpha$. Then, $S$ has one vertex in
$\alpha$ sets $C_i$, and no vertex in the other $K-\alpha$ sets.  We
build a dominating set $T$ in $G'$ as follows: for each vertex in
$S$ we take its copy in $G'$. For each clique $C_i$ without vertices
in $S$, we take $t_i$ and one (anyone) vertex in $C'_i$. The
dominating set~$T$ has size $\alpha+2(K-\alpha)=2K-\alpha$. For each
$C'_i$ there exists a vertex in $T$; so, vertices in $C'_i$, $t_i$
and vertices in $S_i$ are dominated. Now take a vertex in $W_e$ with
$e=(u,v)$, $u\in C_i$ and $v\in C_j$. If $C_i\cap S=\emptyset$ (or
$C_j\cap S=\emptyset$), then $t_i\in T$ (or $t_j\in T$) and, by
construction, $t_i$ is adjacent to all vertices in $W_e$. Otherwise,
there exist $w\in S\cap C_i$ and $x\in S\cap C_j$. Since $S$ is an
independent set, either $w\neq u$ or $x\neq v$. If $w\neq u$, by
construction $w$ (its copy in $C'_i$) is adjacent to all vertices in
$W_e$ and, similarly, for $x$ if $x\neq v$. So, $T$ is a dominating
set.

Conversely, suppose that $T$ is a dominating set of size
$2K-\alpha$. Since $S_i$ is an independent set of size $3K$, we can
assume that $T\cap S_i=\emptyset$ and the same occurs with $W_e$. In
particular, there exists at least one vertex in $T$ in each $C_i$.
Now, suppose that $T$ has two different vertices~$u$ and~$v$ in the
same $C_i$. Then we can replace $v$ by $t_i$ getting a dominating
set (vertices in $S_i$ are still dominated by $u$, and any vertex in
some $W_e$ which is adjacent to $v$ is adjacent to $t_i$). So, we
can assume that $T$ has the following form: exactly one vertex in
each $C_i$, and $K-\alpha$ vertices~$t_i$. Hence, there are $\alpha$
$C'_i$ cliques where $t_i$ is not in $T$. We consider in $G$ the set
$S$ constituted by the $\alpha$ vertices in $T$ in these $\alpha$
sets. Take two vertices $u,v$ in $S$ with, say, $u\in C'_i$ and
$v\in C'_j$ (with $t_i\not \in T$ and $t_j\not\in T$). If there were
an edge $e=(u,v)$ in $G$, neither $u$ nor $v$ would have dominated a
vertex in $W_e$ (by construction). Since neither $t_i$ nor $t_j$ is
in $T$, this set would not have been a dominating set, a
contradiction. So $S$ is an independent set.~\qed
\end{proof}

\begin{theorem}\label{paraminapds}
Under Hypothesis~\ref{hyp0} and \textsc{ETH}, for every
$\epsilon>0$, no approximation algorithm running in time
$f(k)N^{o(k)}$ can achieve approximation ratio smaller than
$2-r-\epsilon$ for \ds{} in graphs of order~$N$.
\end{theorem}
\begin{proof}
In the proof of Theorem~\ref{paraminapis}, we produce a graph~$G_I$ which is made of~$K$ cliques and such that: if at least $(1-\epsilon)m$ clauses are satisfiable in~$I$, then there exists an independent set of size $(1-O(\epsilon))K$; otherwise (at most $(r+\epsilon)m$ clauses are satisfiable in~$I$), the maximum independent set has size at most $(r+O(\epsilon))K$. The previous reduction transforms~$G_I$ in a graph~$G'_I$ such that, applying Lemma~\ref{lemmads}, in the first case there exists a dominating set of size at most $2K-(1-O(\epsilon))K=K(1+O(\epsilon))$ while, in the second case, the size of a dominating set is at least $2K-(r+O(\epsilon))K=K(2-r-O(\epsilon))$. Thus, we get a gap with parameter $k'=K(1+O(\epsilon))$. Note that the number of vertices in~$G'_I$ is $N'=N+K+3K+3K|E_I|=O(N^3)$ (where~$E_I$ is the set of edges in~$G_I$). If we were able to distinguish between these two sets of instances in time $f(k')N'^{o(k')}$, this would allow to distinguish the corresponding independent set instances in time $f(k')N'^{o(k')}=g(k)N^{o(k)}$ since $k'=K(1+O(\epsilon))=k(1+O(\epsilon))$ ($k=K(1-\epsilon^3)$ being the parameter chosen for the graph~$G_I$).~\qed \end{proof}

%Here again, a similar result can be obtained based on Hypothesis 2, leading to an impossibility result for ratios $9/8-\epsilon$ (we get either an independent set of size $(1-\epsilon)K$ or an independent set of size at most $(7/8+\epsilon)K$, so Lemma~\ref{lemmads} creates a gap of $2-(7/8+\epsilon)=9/8-\epsilon$).

Such a lower bound immediately transfers to \textsc{Set Cover}
since a graph on $n$ vertices for \ds{} can be easily transformed
into an equivalent instance of \textsc{Set Cover} with ground
set and set system both of size $n$.
\begin{corollary}\label{paraminapsc}
Under Hypothesis~\ref{hyp0} and \textsc{ETH}, for every
$\epsilon>0$, no approximation algorithm running in time
$f(k)m^{o(k)}$ can achieve approximation ratio smaller than
$2-r-\epsilon$ for \textsc{Set Cover} in instances with~$m$ sets.
\end{corollary}

\section{On the approximability of \is{} and related problems in subexponential
time}\label{secsubexpo}

%{\bf BRUNO'S COMMENT: Eunjung, you added a (strange :-)) paragraph here. I put it in comment in the tex file.}

%The same reductions imply that The second one is that,
%under \textsc{ETH}, \vc{} (where, given a graph~$G(V,E)$, one asks
%for a minimum-size subset~$V'$ of~$V$ such that, for any edge
%in~$E$, at least one of its endpoints is in~$V'$; yet another
%problem among the 21 first \np{}-complete problems
%of~\cite{karpcl}) is not approximable within ratio $7/6 -
%\epsilon$, for any $\epsilon
%> 0$, in time~$2^{n^{1-\delta}}$ for any $\delta > 0$. Let us note that in both (moderately exponential and
%parameterized) approximation frameworks, \vc{} is approximable
%within ratio $r = 2 - \epsilon$, for any $\epsilon <
%1$~\cite{brafer,effapprox}. Finally, based upon the gap-reduction
%of~\cite{lund} for \mcolo{}, we prove that, under \textsc{ETH},
%for any $r>1$ and any $\delta>0$, no $r$-approximation algorithm
%for \mcolo{} can run in time $O(2^{n^{1-\delta}})$ in a graph of
%order~$n$.

As mentioned in Section~\ref{sec:pre}, an almost-linear size {\textsc{PCP}}
construction~\cite{DBLP:conf/focs/MoshkovitzR08} for \3sat allows to get the negative results stated
in Corollaries~\ref{2-query-cor1} and~\ref{2-query-cor2}. In this
section, we present further consequences of Theorem~\ref{2-query},
based upon a combination of known reductions with (almost)
linear size amplifications of the instance.

First, Theorem~\ref{2-query} combined with the reduction
in~\cite{motwa} showing inapproximability results for \is{} in
polynomial time, leads to the following result.
\begin{theorem}\label{subexpis}
Under \textsc{ETH}, for any $r>0$ and any $\delta>0$, there is no
$r$-approximation algorithm for \is{} running in time
$O(2^{N^{1-\delta}})$, where~$N$ is the size of the input graph for
\is{}.
\end{theorem}
\begin{proof}
%\textcolor{blue}
{Given an $\epsilon>0$, let $\epsilon'$ such that
$0<\epsilon'<\epsilon$.
 Given an instance~$\phi$ of
\textsc{3 sat} on $n$ variables, we first apply the sparsification
lemma (with $\epsilon'$) to get $2^{\epsilon' n}$ instances $\phi_i$
on at most $n$ variables. Since each variable appears at most
$c_{\epsilon'}$ times in $\phi_i$, the global size of $\phi_i$ is
$|\phi_i|\leq c_{\epsilon'}n$.}

%\textcolor{blue}
{Consider a particular $\phi_i$, $r>0$ and $\delta>0$. We use the fact that \3sat$\in
{\textsc{PCP}}_{1,r}[(1+o(1))\log |\phi| +D_r,E_r]$ (where~$D_r$ and~$E_r$ are constants that depend only on~$r$), in order to build the following graph~$G_{\phi_i}$ (see also~\cite{motwa}).
%\begin{itemize}
%\item
For any random string~$R$, and any possible value of the~$E_r$ bits read by~V, add a vertex in the graph if~V accepts.
%\item
If two vertices are such that they have at least one contradicting bit (they read the same bit which is~1 for one of them and~0 for the other one), add an edge between them.
%\end{itemize}
In particular, the set of vertices corresponding to the same random
string is a clique.}

%\textcolor{blue}
{Assume that~$\phi_i$ is satisfiable. Then there
exists a proof for which the verifier accepts for any random string~$R$. Take for each random string~$R$ the vertex in~$G_{\phi_i}$ corresponding to this proof. There is no conflict (no edge) between any of these~$2^{|R|}$ vertices, hence $\alpha(G_{\phi_i})=2^{|R|}$ (where, in a graph~$G$,~$\alpha(G)$ denotes the size of a maximum independent set).}

%\textcolor{blue}
{If $\phi_i$ is not satisfiable, then
$\alpha(G_{\phi_i}) \leq r2^{|R|}$. Indeed, suppose that there is an
independent set of size $\alpha>r2^{|R|}$. This independent set
corresponds to a set of bits with no conflict, defining part of a
proof that we can arbitrarily extend to a proof $\Pi$. The
independent set has $\alpha$ vertices corresponding to $\alpha$
random strings (for which~V accepts), meaning that the probability
of acceptance for this proof~$\Pi$ is at least $\alpha/2^{|R|}>r$, a
contradiction with the property of the verifier.}

%\textcolor{blue}
{Furthermore, $G_{\phi_i}$ has $N \leq
2^{|R|}2^{E_r}\leq C'|\phi_i|^{1+o(1)}=Cn^{1+o(1)}$ vertices (for
some constants $C,C'$ that depend on $\epsilon'$) since
$|\phi_i|\leq c_{\epsilon'}n$. Then, one can see that, for any
$r'>r$, an $r'$-approximation algorithm for \is{} running in time
$O(2^{N^{1-\delta}})$ would allow to decide whether $\phi_i$ is
satisfiable or not in time $O(2^{n^{1-\delta'}})$ for some
$\delta'<\delta$. Doing this for each of the formula $\phi_i$ would
allow to decide whether $\phi$ is satisfiable or not in time
$p(n)2^{\epsilon'n}+2^{\epsilon'
n}O(2^{n^{1-\delta'}})=O(2^{\epsilon n})$ (where $p$ is a
polynomial). This is valid for any $\epsilon>0$ so it would
contradicting \textsc{ETH}.~\qed}
\end{proof}

Since (for $k\leq N$),
$N^{k^{1-\delta}}=O(2^{N^{1-\delta'}})$, for some $\delta'<\delta$,
the following result also holds.
\begin{corollary}\label{subexpisbis}
Under \textsc{ETH}, for any $r>0$ and any $\delta>0$, there is no
$r$-approximation algorithm for \is{} (parameterized by $k$)
running in time $O(N^{k^{1-\delta}})$, where~$N$ is the size of
the input graph.
\end{corollary}
The results of Theorem~\ref{subexpis} and Corollary~\ref{subexpisbis} can be
immediately extended to problems that are linked to \is{} by
approximability preserving reductions (that preserve at least
constant ratios) and have linear amplifications of the sizes of the
instances.

For instance, this is the case of \msp{}
(preservation of constant ratios and of ratios functions of the
input size with amplification that is the identity function). This holds for the \cb{} problem where, given a
graph~$G(V,E)$, the goal is to find a maximum-size subset $V'
\subseteq V$ such that the graph~$G[V']$ is a bipartite graph.
%
%{\bf BRUNO'S COMMENT: to save space, the above proof can/should be switched to the appendix.}
%
\begin{proposition}\label{th:cb}
Under \textsc{ETH}, for any $r>0$ and any $\delta>0$, there is no
$r$-approximation algorithm for either \msp{} or
\cb{} running in time $O(2^{n^{1-\delta}})$ in a graph of order~$n$.
\end{proposition}
\begin{proof}
Consider the following reduction from \is{} to \cb{}~(\cite{si}).
Let~$G(V,E)$ be an instance of~\is{} of order~$n$. Construct a
graph~$G'(V',E')$ for \cb{} by taking two distinct copies of~$G$
(denote them by~$G_1$ and~$G_2$, respectively) and adding the
following edges: a vertex~$v_{i_1}$ of copy~$G_1$ is linked with a
vertex~$v_{j_2}$ of~$G_2$, if and only if either $i=j$ or $(v_i,v_j)
\in E$. $G'$ has $2n$ vertices. Let now~$S$ be an independent set
of~$G$. Then, obviously, taking the two copies of $S$ in $G_1$ and
$G_2$ induces a bipartite graph of size $2|S|$. Conversely, consider
an induced bipartite graph in $G'$ of size $t$, and take the largest
among the two color classes. By construction it corresponds to an
independent set in $G$, whose size is at least $t/2$ (note that it
cannot contain 2 copies of the same vertex). So, any $r$-approximate
solution for \cb{} in~$G'$ can be transformed into an
$r$-approximate solution for \is{} in~$G$. Observe finally that the
size of $G'$ is two times the size of $G$. \qed
\end{proof}

Dealing with minimization problems, Theorem~\ref{subexpis}
and Corollary~\ref{subexpisbis} can be extended to \mcolo{}, thanks to the
reduction given in~\cite{lund}.
%In this article, the following
%result is obtained.

Given a graph~$G$ whose vertex set is
partitioned into~$K$ cliques each of size~$S$, and given a prime
number $q>S$, a graph~$H_q$ having the following properties can be
built in polynomial time:
%\begin{itemize}
%\item
(i)~the vertex set of $H_q$ is partitioned into $q^2K$ cliques, each of size~$q^3$;
%\item
(ii)~$\alpha(H_q)\leq \max\{q^2\alpha(G);q^2(\alpha(G)-1)+K;qK\}$;
%\item
(iii)~if $\alpha(G)=K$, then $\chi(H_q)=q^3$.
%(where $\chi(H_q)$ denotes the chromatic number of $H_q$).
%\end{itemize}

Note that this reduction uses the particular structure of graphs
produced in the inapproximability result in~\cite{motwa} (as in
Theorem~\ref{subexpis}). Then, we deduce the following result.
\begin{proposition}\label{th:colo}
Under \textsc{ETH}, for any $r>1$ and any $\delta>0$, there is no
$r$-approximation algorithm for \mcolo{} running in time
$O(2^{n^{1-\delta}})$ in a graph of order~$n$.
\end{proposition}
\begin{proof} Fix a ratio $r>1$, and let $r_{IS}>0$ be such that
$r_{IS}+r_{IS}^2\leq 1/r$. Start from the graph $G_{\phi_i}$
produced in the proof of Theorem~\ref{subexpis} for ratio $r_{IS}$.
The vertex set of $G_{\phi_i}$ is partitioned into  $K=2^{|R|}$
cliques, each of size at most $2^{E_r}$. By adding dummy vertices (a
linear number, since $E_r$ is a fixed constant), we can assume that
each clique has the same size $S=2^{E_r}$, so the number of vertices
in $G_{\phi_i}$ is $N=KS=2^{|R|}2^{E_r}$.

Let $q> \max\{S,1/r_{IS}\}$ be a prime number, and consider the
graph $H_q$ produced from $G_{\phi_i}$ by the reduction
in~\cite{lund} mentioned above. If $\phi_i$ is satisfiable,
$\alpha(G_{\phi_i})=K$ and then by the third property of the graph
$H_q$, $\chi(H_q)=q^3$. Otherwise, by the second property
$\alpha(H_q)\leq
\max\{q^2\alpha(G_\phi);q^2(\alpha(G_\phi)-1)+K;qK\}$.
Formula~$\phi_i$ being not satisfiable, $\alpha(G_{\phi_i})\leq
r_{IS}K$. By the choice of $q$, $qK\leq q^2r_{IS}K$, so
$\alpha(H_q)\leq q^2r_{IS}K+K=(q^2r_{IS}+1)K$. Since the number of
vertices in $H_q$ is $Kq^5$, we get that $\chi(H_q)\geq
q^5/(q^2r_{IS}+1)$. The gap created for the chromatic number in the
two cases is then at least:
$$\frac{q^5}{(q^2r_{IS}+1)q^3}=\frac{1}{r_{IS}+1/q^2}\geq \frac{1}{r_{IS}+r_{IS}^2}\geq r$$
The result follows since $H_q$ has $Kq^5$ vertices and $q$ is a
constant (that depends only on the ratio $r$ and on the constant
number of bits $p$ read by~V), so the size of $H_q$ is linear in the
size of $G_{\phi_i}$.~\qed \end{proof} We consider the
approximability of \vc{} and \minsat{} in subexponential time. The
following statement provides a lower bound to such a possibility.
\begin{proposition}\label{subexpvc}
Under \textsc{ETH}, for any $r>0$ and any $\delta>0$, there is no
$(7/6-\epsilon)$-approximation algorithm for \vc{} running in time
$O(2^{N^{1-\delta}})$ in graphs of order~$N$, nor for \minsat{} running in time $2^{m^{1-\delta}}$ in CNF formul{\ae} with~$m$ clauses.
\end{proposition}
\begin{proof}
We combine Corollary~\ref{2-query-cor1} with the following classical
reduction. Consider an instance $I$ of \textsc{max 3-lin} on~$m$
equations. Build the following graph $G_I$:
\begin{itemize}
\item for any equation and any of the eight possible values of the 3 variables in it, add a vertex in the graph if the equation is satisfied;
\item if two vertices are such that they have one contradicting variable (the same variable has value 1 for one vertex and 0 for the other one), then add an edge between them.
\end{itemize}
In particular, the set of vertices corresponding to the same
equation is a clique. Note that each equation is satisfied by
exactly 4 values of the variables in it. Then, the number of
vertices in the graph is $N = 4m$. Consider an independent set $S$
in the graph $G_I$. Since there is no conflict, it corresponds to a
partial assignment that can be arbitrarily completed into an
assignment~$\tau$ for the whole system. Each vertex in $S$
corresponds to an equation satisfied by $\tau$ (and $S$ has at most
one vertex per equation), so $\tau$ satisfies (at least) $|S|$
equations. Reciprocally, if an assignment $\tau$ satisfies $\alpha$
clauses, there is obviously an independent set of size $\alpha$ in
$G_I$. Hence, if $(1-\epsilon)m$ equations are satisfiable, there
exists an independent set of size at least $(1-\epsilon)m$, i.e., a
vertex cover of size at most $N-(1-\epsilon)m=N(3/4+\epsilon/4)$. If
at most $(1/2+\epsilon)m$ equations are satisfiable, then each
vertex cover has size at least
$N-(1/2+\epsilon)m=N(7/8-\epsilon/4)$.

We now handle \minsat{} problem via the following reduction
(see~\cite{ravi}). Given a graph $G$, build the following instance
on \minsat{}. For each edge $(v_i,v_j)$ add a variable $x_{ij}$. For
each vertex $v_i$ add a clause $c_i$. Variable $x_{ij}$ appears
positively in $c_i$ and negatively in $c_j$. Then, take a vertex
cover $V^*$ of size $k$; for any $x_{ij}$ fix the variable to true
if $v_i\in V^*$, to false otherwise. Consider a clause $c_j$ with
$v_j\not\in V^*$. If $\overline{x_{ij}}$ is in $c_j$ then $v_i$ is
in $V^*$ hence $x_{ij}$ is true; if $x_{ji}$ is in $c_j$ then, by
construction, $x_{ji}$ is false. So $c_j$ is not satisfied, and the
assignment satisfies at most $k$ clauses. Conversely, consider a
truth assignment that satisfies $k$ clauses
$c_{i_1},\cdots,c_{i_k}$. Consider the vertex set
$V^*=\{v_{i_1},\cdots,v_{i_k}\}$. For an edge $(v_i,v_j)$, if
$x_{ij}$ is set to true then $c_i$ is satisfied and $v_i$ is in
$V^*$, otherwise $c_j$ is satisfied and $v_j$ is in $V^*$, so $V^*$
is a vertex cover of size $k$. Since the number of clauses in the
reduction equals the number of vertices in the initial graph, the
result is concluded.\qed
\end{proof}
All the results given in this section are valid under \textsc{ETH}
and rule out some ratio in subexponential time of the form
$2^{n^{1-\delta}}$. It is worth noticing that if ${\textsc LPC}$
holds, then all these result would hold for {\em any} subexponential time.
\begin{corollary}
If ${\textsc LPC}$ holds, under \textsc{ETH} the negative results of
Theorem~\ref{subexpis} and Propositions~\ref{th:cb}, ~\ref{th:colo}
and \ref{subexpvc} hold for any time
complexity $2^{o(n)}$. %(resp. $2^{o(m)}$).
\end{corollary}
\begin{proof}
%\textcolor{blue}
{Using \textsc{LPC}, the same proof as in Theorem~\ref{subexpis}
creates for each $\phi_i$ a graph on $N=O(n)$ variables with either
an independent set of size $\alpha N$ (if $\phi_i$ is satisfiable)
or a maximum independent set of size at most $\alpha/2 N$ (if
$\phi_i$ is not satisfiable). Then using expander graphs,
Theorem~\ref{theoappendix} allows to amplify this gap from~1/2 to
any constant $r>0$ while preserving the linear size of the instance.
Results for the other problems immediately follow from the same
arguments as above.~\qed}
\end{proof}

%{\bf BRUNO'S COMMENT: using ${\textsc LPC}$ we need to use expander graphs to get the result for any ratio. Do you think that the appendix is clear enough? If yes, great. If not, we can remove the appendix and leave the proof like this - very sketchy. Another possibility would be to modify ${\textsc LPC}$ but I think it is better to leave it as it is.}

\section{Conclusion}

This paper presents conditional lower bounds of approximation ratio
in FPT- and subexponential-time. Assuming \textsc{ETH}, we prove
inapproximability in time~$2^{n^{1-\delta}}$ for any $\delta>0$ for
the problems such as: \is{},  \msp{}, \cb{}, \mcolo{}, \vc{}. If
{\textsc Linear PCP Conjecture} turns out to hold, even in
time~$2^{o(n)}$ we cannot approximate any better. Also assuming
\textsc{ETH}, we proved that \textsc{Linear PCP Conjecture} implies
 FPT-time inapproximabilty of \is{} (for any ratio) and \ds{} (for some ratio).

Our effort in this paper is only a first step and we wish to motivate further research. There remains a range of problems to be tackled, among which we propose the followings.
\begin{itemize}
\item Our inapproximability results, in particular those in FPT-time, are conditional upon {\textsc Linear PCP Conjecture}. Is it possible to relax the condition to a more plausible one?
\item Or, we dare ask whether (certain) inapproximability results in FPT-time imply strong improvement in PCP theorem. For example, would the converse of Lemma \ref{lpc2hyp} hold?
\end{itemize}
%
%t-inapproximabilityof the following problemsAs we have seen, the condition interesting butsome new insight in subexponential and fixed parameter approximation by For negative results in FPT time (Section~\ref{secparam}), similar results can be obtained using slightly relaxed versions of Hypotheses~\ref{hyp1} and~\ref{hyp2}. % (such as same hypothesis %but with other ratios).
%As mentioned in the introduction, it is worth noticing that a stronger version of {\textsc{PCP}}, for instance, a {\textsc{PCP}} using $\log n$ bits instead of $(1+o(1))\log n$\footnote{More precisely, a result saying that for every $\epsilon>0$, \3sat$\in \textsc{PCP}_{1,\epsilon}[\log n +D_\epsilon,E_\epsilon]$ for some constant $D_\epsilon$ and $E_\epsilon$.}, would immediately lead to the following:
%\begin{itemize}
%    \item results in Section~\ref{secsubexpo} would be valid for {\it     any} subexponential time $2^{o(n)}$;
%    \item results in Section~\ref{secparam} would be valid under     \textsc{ETH};
%\end{itemize}
%thus definitely solving the open questions mentioned in introduction.
%
Note that we have considered in this article constant
approximation ratios. In this sense, Theorem~\ref{subexpis} is
``tight'' with respect to approximation ratios since, as mentioned
in Section~\ref{sec:pre}, ratio $1/r(n)$ is achievable in
subexponential time for any increasing and unbounded function $r$.
However, dealing with parameterized approximation algorithms,
achieving a non constant ratio is also an open question. More
precisely, finding in FPT-time an independent set of size $g(k)$
when there exists an independent set of size $k$ is not known for
{\it any} unbounded and increasing function $g$.

Finally, let us note that, in the same vein of our
work,~\cite{mathieson_pcp} in his recent paper initiates a proof
checking view of parameterized complexity, by proposing a
parameterized PCP and by giving a parameterized PCP characterization
of~W[1]. Possible links between these two approaches are worth being
investigated in future works.

%\bibliographystyle{plain}
%\bibliography{biblio}

%\newpage

%\appendix

%\section{Gap amplification}\label{appendixgraphamplif}

\end{document}